\documentclass[11pt,a4paper,english]{article}

\title{Relativistic Corrections to the Moyal-Weyl Spacetime}
\author{A. Much 
\\ \footnotesize{Faculty of Mathematics, University of Vienna, 1090 Vienna, Austria}}

\usepackage[pdftex]{graphicx} % Einbindung von Grafiken (pdf, png, jpg)
\usepackage{float}            % bietet Option [H] für bombenfestes Verankern
\usepackage[]{babel}
\usepackage[T1]{fontenc}   
\usepackage[cmex10]{amsmath}  % für erw. Formeloptionen, Option [] zur Vermeidung von Type3-Fonts
\usepackage{amstext}          % für Klartext via \text{} in Formeln
\usepackage{mathrsfs}
\usepackage{amsfonts}         % für komplexere Formeln (Mengensymbole ...)
\usepackage{amssymb}          % für komplexere Formeln (Mengensymbole ...)
\usepackage{setspace}
\usepackage{amsthm}
\usepackage{amsmath}
\usepackage{bm}               % bold math, für \bm{}
\usepackage{enumerate}        % verbessert Aufzählungen
\usepackage[bottom]{footmisc} % Fussnoten am Seitenende
\usepackage{array}            % für Tabellen: bindet tabular-Umgebung ein
\usepackage{textcomp}         % für \textdegree , \textcelsius , macht aber manchmal Probleme (!!)
\usepackage{pdfpages}         % für die Einbindung kompletter pdf-*Seiten*
\usepackage{parskip}          % zw. Absätzen: eine knappe Leerzeile statt hängender Einzüge
\usepackage[right]{eurosym}   % Eurosymbol
\usepackage{xcolor}           % farbiger Text
\usepackage[square,numbers]{natbib}	  % Für \setlength{\bibsep}{3mm}; square macht eckige

\definecolor{darkred}{rgb}{0.7,0.0,0.0}

\newtheorem{theorem}{\textsc{Theorem}}[section]
\newtheorem{lemma}{\textsc{Lemma}}[section]

\theoremstyle{definition}
\newtheorem{definition}{\textsc{Definition}}[section]

\theoremstyle{remark}
\newtheorem{remark}{Remark}[section]

\usepackage{fancyhdr}
 \numberwithin{equation}{section} 
\begin{document}
\maketitle
\abstract{We use the framework of quantum field theory   to obtain  by deformation the  Moyal-Weyl spacetime.   This idea is extracted from recent progress in deformation theory concerning the emergence of the quantum plane of the Landau-quantization. The quantum field theoretical  emerging  spacetime is  not equal to the standard Moyal-Weyl plane, but terms resembling relativistic corrections occur.} 
%%%%%%%%%%%%%%%%%%%%%%
% Inhaltsverzeichnis %
  \tableofcontents
%%%%%%%%%%%%%%%%%%%%%%

\section{Introduction}
The most prominent and researched example of quantized spacetimes is the case of constant noncommutativity, the so called Moyal-Weyl plane.
 Where did we get this plane from? Is it possible to derive it from anything we know, or is it just another experimentally non-realizable idea of theoretical physics?  \\\\  From a mathematical point of view, one might argue that the Moyal-Weyl plane is just  the coordinate generalization of the quantum-mechanical symplectic structure between the momentum and the coordinate.  From a  physical point of view and guided by simplicity, nontrivial  commutation relations of the coordinates emerge  from the fact  that localization with high accuracy at Planck scale causes a gravitational collapse, \cite{DFR}.   The  commutation relations found  are identified with those of the Moyal-Weyl.   \\\\
 Another {beautiful} physical point of view is the emergence of such a plane in the Landau-quantization. In this  effect the physical constant describing noncommutativity of  space is given by the inverse of a magnetic field, \cite{EZ, SZ}.\\\\ Nonetheless, a conceptual  mathematical derivation of such a spacetime would lead to  a deeper understanding of the physical nature of quantum-spacetimes.  Such a concept was developed for the quantum-mechanical (QM) case in \cite{Muc1}. It was done in a rigorous mathematical fashion by using the deformation technique known under warped convolutions \cite{BLS}. In the QM-case, we were able to identify the quantum plane of the Landau quantization by taking the commutator of the deformed coordinate operator. The generators used for the deformation of the coordinate operator are the momentum operators.   These techniques were also used rigorously in quantum measurement theory \cite{AA}.
\\\\ The idea developed in the QM-case is extended in this paper to a relativistic quantum field theoretical context. This is done by using the creation and annihilation operators of the free scalar field. Hence, we first define a quantum field theoretical version of an operator that is conjugate to the second-quantized momentum.
 This is done by taking the QM-definition using the momentum operator and performing a second-quantization onto the bosonic Fock space. Moreover, the definition is   extended to a temporal component and the deformed commutator of these operators is calculated. \\\\
 Now we recognized with regards to the spatial part of the conjugate operator an equivalence to the Newton-Wigner-Pryce (NWP) operator. In the context of relativistic particles, the NWP  operator is usually
mentioned as the rightful position  or center-of-mass operator. To obtain this  operator,
 certain physical requirements where imposed on localized states, \cite{NW49}, while in
\cite{PR48}  the same object was found by a relativistic generalization of the Newtonian definition of the
mass-center. However, localization in terms of the position operator is beset by problems regarding relativity. It is in clear conflict with relativistic covariance and causality in quantum physics. This fact can be verified by using the spectral theory of self-adjoint operators (see for example \cite[Theorem1]{Yng14}).\\\\
An  exit strategy concerning localization is to shift the focus to localized fields, where the spacetime coordinates are common variables of such fields. In this context, it is worth mentioning that a quantum field deformed with the coordinate operator belongs to the class of wedge-local fields (\cite[Proposition 5.4]{Muc2}). Wedge-locality is a weaker form of the common point-wise locality and it seems the appropriate notion of causality in noncommutative QFT, \cite{GL1, Sol}. \\\\ Although  the NW-localization violates relativistic covariance and causality, it holds at least in an approximate sense for distances of the order of the Compton wave length or smaller such as the Planck-length, \cite{Mund05} or \cite[Chapter 3c]{Sch}.  And it is precisely around the Planck-length scale where  noncommutative spacetimes are supposed to be non-negligible, \cite{LS, SZ}.
\\\\
Hence, following the QM-case the paper is organized as follows:  We first review the most important definitions of  warped convolutions. In Section \ref{s3}, we define the second quantized coordinate operator and prove  its equivalence to the NWP-operator. Moreover, we show deformations of such unbounded operators are mathematically well-defined. By calculating the deformed commutator of the coordinate operators, we yield terms resembling relativistic corrections to the standard Moyal-Weyl. 
 \section{Preliminaries }\label{s2}

\subsection{Warped Convolutions}
In  this section we write all   basic definitions and lemmas of the deformation known under the name of warped convolutions. For  proofs of the respective lemmas we refer the reader to the original works \cite{BLS, GL1}. \newline\newline 
The authors start their investigation with a $C^{*}$-dynamical system $(\mathcal{A},\mathbb{R}^d)$, \cite{P97}.
It consists of a $C^{*}$-algebra $\mathcal{A}$  equipped with a strongly continuous automorphic action of the group $\mathbb{R}^d$ which will be denoted by $\alpha$.
Furthermore, let  $\mathcal{B}(\mathscr{H})$ be the Hilbert space of bounded
operators on $\mathscr{H}$ and let  the adjoint action $\alpha$ be implemented  by the weakly continuous unitary representation $U$. Then,  it is argued that  since the unitary representation $U$ can be extended to the algebra $\mathcal{B}(\mathscr{H})$,   there is no lost of generality   when one proceeds to  the $C^{*}$-dynamical system  $(C^{*}\subset  \mathcal{B}(\mathscr{H}),\mathbb{R}^d)$.  Here $C^{*}\subset  \mathcal{B}(\mathscr{H})$  is the $C^{*}$-algebra of all operators on which $\alpha$ acts strongly continuously.
\newline\newline 
Hence, we start by assuming the
existence of a strongly continuous unitary group $U$ that is a representation of the additive
group $\mathbb{R}^{d}$, $d\geq2$, on some separable Hilbert space $\mathscr{H}$. 
Moreover, to define the deformation of operators belonging to a $C^{*}$-algebra $C^{*}\subset  \mathcal{B}(\mathscr{H})$, we consider
elements belonging to the subalgebra $C^{\infty}\subset C^{*} $. The subalgebra $C^{\infty}$ is
defined to be  the $*-$algebra of smooth elements (in the norm topology) with respect to $\alpha$, which
is the adjoint action of a weakly continuous
unitary representation $U$ of $\mathbb{R}^{d}$ given by  
\begin{equation*}
 \alpha_{x}(A)=U(x)\,A\,U(x)^{-1}, \quad x \in \mathbb{R}^{d}.
\end{equation*} 
Let $\mathcal{D}$ be the dense domain of vectors in $\mathscr{H}$ which transform
smoothly under the adjoint action of $U$.
Then, the warped convolutions  for operators  $A\in C^{\infty}$ are given by
the following definition.\\

\begin{definition}
 Let $\theta$ be a real skew-symmetric matrix on $\mathbb{R}^{d}$ w.r.t. the Minkowski metric, let $A\in C^{\infty}$  and let $E$ be
the spectral resolution of the unitary operator $U$.  Then, the corresponding \textbf{warped convolution} $A_{\theta}$
of $A$ is defined on the dense domain $\mathcal{D}$ according to
\begin{equation}\label{WC}
 A_{\theta }:=\int \alpha_{\theta x}(A)dE(x),
\end{equation}
where $\alpha$ denotes the adjoint action of $U$ given by $\alpha_{k}(A)=U(k)\,A\,U(k)^{-1}$.
\end{definition} 
The restriction in the choice of operators is owed to the fact that the deformation is performed
with operator valued integrals. Furthermore,  one can represent the warped
convolution of $
 {A} \in {C}^{\infty}$    by $\int \alpha_{\theta x}(A) dE(x)$, on the dense domain
$\mathcal{D}\subset\mathscr{H}$ of vectors smooth w.r.t. the action of $U$,  in terms of
strong limits 
\begin{equation*}
\int\alpha_{\theta x}(A) dE(x)\Phi=(2\pi)^{-d}
\lim_{\epsilon\rightarrow 0}
\iint  dx\, dy \,\chi(\epsilon x,\epsilon y )\,e^{-ixy}\, U(y)\, \alpha_{\theta x}(A)\Phi,  
\end{equation*}
where $\chi \in\mathscr{S}(\mathbb{R}^{d}\times\mathbb{R}^{d})$ with $\chi(0,0)=1$. In an intermediate step of showing this equivalence it was as well proven that $\mathcal{D}$ is \textbf{stable} under the deformation. 
The latter representation makes calculations and proofs concerning the existence of  integrals
easier. 
During this work we make use of both representations.  \\\\
In the following lemma we introduce the deformed product, known as the Rieffel product \cite{RI},
by using warped convolutions. The circumstance that the two are interrelated is due to
the fact that  warped convolutions supply isometric
representations of Rieffel's strict deformations of $C^{*}$-dynamical systems with
actions of $\mathbb{R}^{d}$.  The definition of the Rieffel product, given by warped convolutions, is used in the next chapter to calculate
the deformed commutators. \newline
\begin{lemma}\label{dpl}
Let $\theta$ be a real skew-symmetric matrix on $\mathbb{R}^{d}$  w.r.t. the Minkowski metric, $  {A},   {B} \in
 C^{\infty}$  and let $\Phi\in\mathcal{D}$. Then, the product of two deformed operators $A,   B$ is given by
\begin{equation*}
 {A}_{\theta}   {B}_{\theta} \Phi= (A\times_{\theta}B)_{\theta}\Phi,
\end{equation*}
where the \textbf{deformed product} $\times_{\theta}$ is  the Rieffel product
on $ {C}^{\infty}$ defined as, 
\begin{equation}\label{dp0}
(A\times_{\theta}B )\Phi=(2\pi)^{-d}
\lim_{\epsilon\rightarrow 0}
\iint dx\, dy\,\chi(\epsilon x,\epsilon y )e^{-ixy} \alpha_{\theta x}(A)\alpha_{y}(B)\Phi.
\end{equation}
\end{lemma}
Next, we shall give the definition of a deformed commutator using the deformed product. 

\begin{definition}\label{dpl1}
Let $\theta$ be a real skew symmetric-matrix on $\mathbb{R}^{d}$, $  {A},   {B} \in
 {C}^{\infty}$ and let $\Phi\in\mathcal{D}$. Then, the \textbf{deformed commutator} is  defined as,
\begin{equation}\label{dp1}
[A\stackrel{\times_{\theta}}{,} B ]\Phi:= (A\times_{\theta}B-B\times_{\theta}A )\Phi.
\end{equation}
\end{definition}
The Rieffel product, i.e. an associative product,  was used for the definition  and hence the commutator has the standard properties. Furthermore, it is interesting to note that in \cite{BC}, a similar commutator was given and dubbed the Q-commutator.  It differs from our current definition only by an additional term, which is the commutative product of the two operators $  {A},   {B} \in
 {C}^{\infty}$.
\\\\ As a matter of fact, we intend to deform \textbf{unbounded operators} and hence we are obliged to prove  that the deformation
formula, given as an oscillatory integral, is well-defined.   This is done in Section \ref{s3}. 
\subsection{ Bosonic Fock 
space} The ($d=n+1 $)-dimensional relativistic  bosonic Fock 
space is defined in the following. A particle with momentum $\mathbf{p} \in \mathbb{R}^n$ has in   the massless case the energy $\omega_{\mathbf{p}}=+\sqrt{\mathbf{p}^2}$. Moreover the Lorentz-invariant measure is given by   $d^n\mu(\mathbf{p} )=d^n\mathbf{p}/( {2\omega_{\mathbf{p}}})$.
\begin{definition}\label{bf}
 The \textbf{bosonic Fock space} $\mathscr{F^{+}({H})}$ is defined 
as in \cite{BR}:
\begin{equation*}
\mathscr{F^{+}({H})}=\bigoplus_{k=0}^{\infty}\mathscr{H}_{k}^{+},
\end{equation*}
where $\mathscr{H}_{0}=\mathbb{C}$ and  the symmetric $k$-particle subspaces are given as
  \begin{align*}
\mathscr{H}_{k}^{+} =\{\Psi_{k}: \underbrace{\partial V^{+}  \times  \dots \times 
\partial V^{+}}_{k-times} \rightarrow \mathbb{C}\quad \mathrm{symmetric}
| \left\Vert  \Psi_k \right\Vert^2  <\infty\},
\end{align*}
with $
  \partial V^{+}:=\{p\in \mathbb{R}^{d}|p^2=0,p_0>0\}.$
  \end{definition}
  The particle annihilation and creation operators $a,a^{*}$ of the  bosonic
Fock space satisfy the following commutator relations
  \begin{align}\label{pccr}
[a_c(\mathbf{p}), a_c^{*}(\mathbf{q})]=2\omega_{\mathbf{p} 
}\delta^n(\mathbf{p}-\mathbf{q}), \qquad
[a_c(\mathbf{p}), a_c(\mathbf{q})]=0=[a_c^{*}(\mathbf{p}), a_c^{*}(\mathbf{q})].
\end{align} 
By using  $a_c,a_c^{*}$   the particle number operator and the momentum operator are  defined in the following manner.
 \begin{equation}\label{pcaopm}
N=\int d^n\mu(\mathbf{p}) a_c^{*}(\mathbf{p})a_c(\mathbf{p}), \qquad P_{\mu}=\int
d^n\mu(\mathbf{p})\,  p_{\mu} a_c^{*}(\mathbf{p})a_c(\mathbf{p}),
\end{equation}
where  $p_{\mu}=(\omega_{\mathbf{p}}, \mathbf{p})$. Moreover, we shall use  the terms of the annihilation and creation operators   in the noncovariant  representation,
\begin{equation*} 
 {a}(\textbf{p}):= \frac{{a}_c (\textbf{p})}{\sqrt{2\omega_{\mathbf{p}}}},\qquad  {a}^* (\textbf{p}):= \frac{{a}_c^* (\textbf{p})}{\sqrt{2\omega_{\mathbf{p}}}}.
\end{equation*} 
Throughout this work we use   $\mu, \,\nu=0,\dots,n$ and we shall use Latin letters for the spatial components which run from $1,\dots,n$. Moreover, we choose the following convention for the Minkowski scalar product of $ d $-dimensional vectors, $a\cdot b=a_0b^0+a_kb^k=a_0b^0- \vec{a}\cdot\vec{b}$.
\section{QFT-Moyal-Weyl from Deformation}\label{s3}
In the context of QM the deformation of the coordinate operator  with the
momentum operator gave as the quantum plane of  the Landau-quantization, (see \cite[Lemma 4.3]{Muc1}).
Since the commutator commutes with the generators of deformation, one can view the result as 
the deformed commutator of the coordinate operators, (see Definition \ref{dpl1}). This is the point of view taken in this paper. 
\\\\ Hence, in the present section
we follow the idea found in a QM context and calculate the  deformed commutator of the \textbf{quantum field theoretical}
conjugate operator  by using the momentum   for deformation. The resulting  quantum spacetime is called the \textbf{QFT-Moyal-Weyl spacetime}.
\subsection{Second-quantized Coordinate Operator}
 In the so called \textbf{Schr\"odinger
representation}, \cite{RS1, BEH} the pair of operators
$(P_{i},Q_{j})$, satisfying the \textbf{canonical commutation relations}  (CCR)
\begin{equation}\label{ccr}
[Q_i,P_{k}]=-i\eta_{ik},
\end{equation}
are represented as essentially self-adjoint operators on the dense domain
$\mathscr{S}(\mathbb{R}^n)$. Here $Q_{i}$ and $P_{k}$ are the closures of
$x_{i}$ and multiplication by $i {\partial}/{\partial x^k}$ on
$\mathscr{S}(\mathbb{R}^n)$ respectively. 
Now in order to give the spatial coordinate operator in a \textbf{QFT-context} we take the quantum mechanical unitary equivalence of the momentum and coordinate operator into account and perform a second quantization.
\begin{lemma}\label{xsa}
Let $U_{\mathscr{F}}$ be the unitary operator of the $n$-dimensional Fourier transform, which is given by the following action on a function $\Psi\in\mathscr{S}(\mathbb{R}^n)$,

\begin{equation}\label{ft}
(U_{\mathscr{F}}\Psi)(p):=  {(2\pi)^{-n/2}}\int\limits d^n\mathbf{x}\,\,
e^{ ip_kx^k}\Psi(x).
\end{equation}  
 Then, the Fourier transformation gives the unitary equivalence of operator $\mathbf{Q}$  to the operator $\mathbf{P}$ as follows,
\begin{equation}\label{uexp}
Q_{j}=U_{\mathscr{F}}^{-1}P_{j}U_{\mathscr{F}}.
\end{equation}
 $\mathbf{Q}$ is an essentially self-adjoint operator on the dense domain $ D(\mathbf{Q})=
U_{\mathscr{F}}^{-1}D(\mathbf{P})=U_{\mathscr{F}}^{-1}\mathscr{S}(\mathbb{R}^n)=\mathscr{S}
(\mathbb{R}^n)$. Furthermore, its second-quantization $d\Gamma(.)$ is an essential self-adjoint operator on the dense domain $D(\mathbf{Q})_{\otimes}\subset \mathscr{F ({H})}$, which is the set of $\psi=\{\psi_{0},\psi_{1},\cdots\}$ such
that
$\psi_{k}=0$ for $k$ large enough and $\psi_{k}\in \bigotimes_{i=1}^k D(\mathbf{Q})$ for each $k$. The second-quantization of $\mathbf{Q}$ is given as
\begin{equation}\label{uexp2}
X_j:= d\Gamma(Q_{j})=\Gamma(U_{\mathscr{F}}^{-1})P_{j}\Gamma(U_{\mathscr{F}}),
\end{equation} with $\Gamma(U_{\mathscr{F}})=\bigoplus_{k=0}^{\infty} U_{\mathscr{F}}^{\otimes k}$. 
\end{lemma}
\begin{proof}
 The theorem for the one-particle subspace can be found in \cite[Equation 7.4]{BEH}.
Now  the self-adjointness of the second quantized operator can be found in \cite[Theorem VIII.33]{RS1} and \cite[Chapter X.7]{RS}. By a slight abuse of notation we use the same symbol for the second-quantized momentum operator.
\end{proof}

\begin{remark}\label{ru}
 Note that the quantum mechanical equivalence is demonstrated for a one-particle subspace w.r.t. Lebesgue measure.  To change the Lebesgue   to the Lorentz-invariant measure of QFT one combines the former with a multiplication operator. i.e. $(2P_0)^{-1}$. Hence, to have unitarity of the Fourier transformation w.r.t. the  Lorentz-invariant measure we combine it with the unitary  multiplication operator as follows. We consider the adjoint action of the operator $U_{\mathscr{F}}$ on the annihilation operator smeared  with a function $f \in\mathscr{H}_1 $,  
\begin{align*}
(\Gamma(U_{\mathscr{F}})a (\overline{f})\Gamma(U_{\mathscr{F}})^{-1})&=
\int \frac{d^n \mathbf{p}}{\sqrt{2\omega_{\mathbf{p}}}}
f(\mathbf{p})(\Gamma(U_{\mathscr{F}}) {a} (\mathbf{p})\Gamma(U_{\mathscr{F}})^{-1})
\\&=
\int \frac{d^n \mathbf{p}}{\sqrt{2\omega_{\mathbf{p}}}}
f(\mathbf{p}) \left( { {U}}_{\mathscr{F}} {a} \right)(\mathbf{p})
\\&=
\int \frac{d^n \mathbf{p}}{\sqrt{2\omega_{\mathbf{p}}}}\underbrace{\left(
\sqrt{2\omega} {U}_{\mathscr{F}}\frac{1}{\sqrt{2\omega}}
f\right)}_{=: {U}_{\mathscr{F},\omega}f}(\mathbf{p})  {a} (\mathbf{p}) = a(\overline{{U}_{\mathscr{F},\omega}f} ).
\end{align*}
The operator ${U}_{\mathscr{F},\omega}$ is unitary on the one-particle subspace w.r.t. the Lorentz-invariant measure.

\end{remark}

\begin{lemma}\label{xpr}
The coordinate operator has the following bosonic Fock space representation  
\begin{equation}\label{uexp3}
X_j=-i\int d^n \mathbf{p}\,  {a}^*(\textbf{p}) \frac{\partial}{\partial p^j}  {a}(\textbf{p}).
\end{equation}
Moreover, let the dense domain  $D(\mathbf{Q})_{\otimes}^{+}\subset \mathscr{F^{+} ({H})}$ be the projection onto the symmetric subspace of $D(\mathbf{Q})_{\otimes}$. Then, the second-quantized operator in the bosonic Fock space representation is an essential self-adjoint operator on $D(\mathbf{Q})_{\otimes}^{+}$.
\end{lemma}

\begin{proof} The proof for the Fock space representation (see \cite{SS}) can be found in \cite[Theorem 5.1]{Muc2}. Nevertheless, we outline it in more detail to ease to the reader the proof of the next somewhat more complicated lemma. 
 Since we intend to represent the second quantized coordinate operator on the symmetric Fock space, we use the momentum operator given in Equation (\ref{pcaopm}) and calculate the unitary equivalence explicitly on the creation and annihilation operators. Also note that the second quantized operator of the Fourier transformation commutes with the projections on to the symmetric Fock space (see for example \cite[Chapter VIII.10, Example 2]{RS1}). Therefore, we calculate the action of $\Gamma(U_{\mathscr{F}}^{-1})$ directly on the momentum operator in its  symmetric Fock space representation, i.e. 
\begin{align*} 
 X_j&=\Gamma(U_{\mathscr{F}}^{-1})P_{j}\Gamma(U_{\mathscr{F}})=  \int
d^n \mathbf{p}\,  p_{j}  {(\hat{U}^{-1}_{\mathscr{F}}{a} )}^{*}(\mathbf{p})(\hat{U}^{-1}_{\mathscr{F}}{a})(\mathbf{p})\\&= 
\int
d^n \mathbf{p}\,{a}^{*}(\mathbf{p})\left(\hat{U}_{\mathscr{F}} p_{j}  \underbrace{(\hat{U}^{-1}_{\mathscr{F}}{a} )}_{=:\tilde{a}} \right)(\mathbf{p}) \\&= \int
d^n \mathbf{p}\,{a}^{*}(\mathbf{p})\left(
(2\pi)^{-n/2} \int d^n\mathbf{x}\,e^{ip_k x^k}x_j\,\tilde{a}(\mathbf{x})
 \right) \\
&= -i(2\pi)^{-n/2}\int
d^n \mathbf{p}\,{a}^{*}(\mathbf{p})\left(
\frac{\partial}{\partial p^j} \int d^n\mathbf{x}\,
e^{ip_k x^k} \,\tilde{a}(\mathbf{x})
 \right)  \\&=   -i \int
d^n \mathbf{p}\,{a}^{*}(\mathbf{p}) 
\frac{\partial}{\partial p^j} {a} (\mathbf{p}), 
\end{align*}
where for the Fourier transform we used \cite[Lemma IX.1]{RS}. The essential self-adjointness of the second-quantized operator on $D(\mathbf{Q})_{\otimes}^{+}$  follows from 
\cite[Theorem VIII.33]{RS1}.
\end{proof} 
\begin{remark}\label{fu}
The proof for the representation of the coordinate operator was done by using the unitarity of the Fourier transformation w.r.t. the Lebesgue measure, since the momentum operator was given in the noncovariant representation.  However, the proof can be done analogously by using the momentum operator in covariant representation and the operator ${U}_{\mathscr{F},\omega}$  for the equivalence. This is merely a question of covariant versus noncovariant representation.
\end{remark}

Which quality of $\mathbf{X}$ makes it an appropriate coordinate operator? Well, first of all it is an essentially self-adjoint operator which is the second quantization of the QM-coordinate. Moreover in the following we stage the equivalence to the Newton-Wigner-Pryce operator (NWP). The NWP operator can  be given by the product of generators
of the Poincar\'e group in the following way, \cite{  BAC2, BERG65, CA65, J80}
\begin{equation*}
Q_{j}^{NWP}= \frac{1}{2P_0}M_{0j}+M_{0j}\frac{1}{2P_0}.
\end{equation*} 
On a one particle wave function the position operator is given by the following action,
\cite[Chapter 3c, Equation 35]{Sch}
\begin{equation}\label{NWP}
(Q_{j}^{NWP}\varphi)(\mathbf{p})=-i \left( \frac{p_j}{2\omega_{\mathbf{p}}^2}
+   \frac{\partial}{\partial p^j } 
\right)\varphi(\mathbf{p}).
\end{equation}
 
\begin{lemma}\label{enwp}
The second quantization of the \textbf{Newton-Wigner-Pryce}  operator onto the bosonic Fock space, defined on the dense domain  $D(\mathbf{Q})_{\otimes}^{+}\subset \mathscr{F^{+} ({H})}$,  is given in terms of the covariant annihilation and creation operators as follows, 
\begin{equation*}
 X_{j}^{NWP} =
 -i \int
d^n \mu(\mathbf{p})\,
 {{a}^{*}_c (\textbf{p})} 
  \left(\frac{p_j}{2\omega_{\mathbf{p}}^2} 
   +  \frac{\partial}{\partial p^j} \right){a}_c (\textbf{p}).
\end{equation*}
Moreover, it is equivalent to the position
operator $X_j$ given in Equation (\ref{uexp3}).
\end{lemma}
\begin{proof}
The first statement is fairly obvious by taking the action of the NWP operator (see Equation (\ref{NWP})) on the the dense domain $ D(\mathbf{Q})$ and performing a second quantization onto the bosonic Fock space. Now the equivalence can be as well  easily seen by rewriting the spatial conjugate operator in terms of the covariant annihilation and creation operators, i.e. 
\begin{align*} 
 X_j&=     -i \int
d^n \mathbf{p}\,{a}^{*}(\mathbf{p}) 
\frac{\partial}{\partial p^j} {a} (\mathbf{p}) \\&=
 -i \int
d^n \mathbf{p}\,
\frac{{a}^{*}_c (\textbf{p})}{\sqrt{2\omega_{\mathbf{p}}}}
\frac{\partial}{\partial p^j}  \left(
\frac{{a}_c (\textbf{p})}{\sqrt{2\omega_{\mathbf{p}}}}\right)\\&=
 -i \int
d^n \mu(\mathbf{p})\,
 {{a}^{*}_c (\textbf{p})} 
  \left(\frac{p_j}{2\omega_{\mathbf{p}}^2} 
   +  \frac{\partial}{\partial p^j} \right){a}_c (\textbf{p}),
\end{align*} 
Hence, the second quantization of the NWP operator, given by $d\Gamma(X_{j}^{NWP})$ onto the bosonic Fock space, is equivalent to $X_j$.
\end{proof} 
 A few comments are in order. First, the Newton-Wigner-Pryce operator is often given as the
product of generators of the Poincar\'e group. This representation is only true for one-particle
states and not for $n$-particle states. The reason   is the simple observation that the product of second quantized
operators is not equal to the second quantized product of the operators, i.e.  $d\Gamma(M_{0j}
P_{0}^{-1})\neq d\Gamma(M_{0j})d\Gamma(P_{0}^{-1})$.   Therefore, the representation of the
NWP-operator as the product of the boost operator and the inverse of the Hamiltonian must be
discarded for an $n$-particle system.\newline\newline Second, from Lemma \ref{enwp} it follows
that the Newton-Wigner-Pryce
operator is equivalent to the position operator that we also obtained by   second quantization
of the quantum mechanical  position operator.   One reason why this fact may have not been
obvious, is owed to the representation of the operator. The second quantized coordinate operator  is
given in a non-covariant manner, while the NWP-operator is given in a covariant fashion. Of
course, the
difference of representation is merely a normalization feature and thus for the physical
interpretation, and specially for the second quantization not relevant.
 \\\\ Since in QFT we work in a relativistic setting we   \textbf{define} the zero component of the coordinate operator in the same fashion as the spatial part, i.e. 
\begin{equation}\label{x0}
 X_0:=\Gamma(U_{\mathscr{F}}^{-1})\,P_{0}\,\Gamma(U_{\mathscr{F}}).
\end{equation}
Note that the domain of essential self-adjointness of $P_0$ on the one-particle subspace is $\mathscr{S}
(\mathbb{R}^n)$. 
\begin{lemma} The operator  $X_0$   is given as an essentially self-adjoint operator on the dense domain    $ D(\mathbf{Q})_{\otimes}\subset\mathscr{F ({H})}$  and has the following bosonic Fock space representation  
\begin{equation}\label{uexpm3}
X_0=c_{d}
 \int
d^n \mathbf{p} \, 
{a}^*(\mathbf{p})   \left({\omega}^{ -{d}}  \ast
 {a}\right) (\mathbf{p} ),
\end{equation}
with  constant  $c_{d}:=-\pi^{-\frac{d}{2}}\Gamma(\frac{d}{2})$ and $ \ast$ denoting the convolution. Moreover, the second-quantized operator $X_{0}$ in the bosonic Fock space representation is an essential self-adjoint operator on $D(\mathbf{Q})_{\otimes}^{+}$.
 \end{lemma}
\begin{proof}As for the spatial part the argument of essential self-adjointness is based on the unitary equivalence, given by the second quantized Fourier transformation, to the essential self-adjoint operator $P_0$. Let us now turn to the result in terms of the massless  creation and annihilation operators, 
\begin{align*} 
 X_0&=\Gamma(U_{\mathscr{F}}^{-1})P_{0}\Gamma(U_{\mathscr{F}})= \int
d^n \mathbf{p}\,  \omega_{\mathbf{p}}  {(\hat{U}^{-1}_{\mathscr{F}}{a} )}^{*}(\mathbf{p})(\hat{U}^{-1}_{\mathscr{F}}{a})(\mathbf{p})\\&=  \int
d^n \mathbf{p}\,{a}^{*}(\mathbf{p})\left(\hat{U}_{\mathscr{F}} {|\mathbf{p}|}{\tilde{a}}  \right)(\mathbf{p}) \\& =\int
d^n \mathbf{p}\,{a}^{*}(\mathbf{p})\left(
(2\pi)^{-n/2} \int d^n\mathbf{x}\,e^{ip_k x^k}  {|\mathbf{x}|}\,\tilde{a}(\mathbf{x})
 \right)
\\
 &=(2\pi)^{-n}
\int
d^n \mathbf{p}\,{a}^{*}(\mathbf{p})\left(
 \int d^n\mathbf{x}\,e^{ip_k x^k} {|\mathbf{x}|}\, \int d^n\mathbf{q} \,e^{-iq_k x^k} {a}(\mathbf{q})
 \right)\\ &=(2\pi)^{-n}
\iint
d^n \mathbf{p} \,d^n\mathbf{q}\,{a}^{*}(\mathbf{p})\left(
 \int d^n\mathbf{x}\,e^{i(p-q)_k x^k}{|\mathbf{x}|}\,   \right)  {a}(\mathbf{q})\\&=-\pi^{-\frac{n+1}{2}}\Gamma(\frac{n+1}{2})
\iint
d^n \mathbf{p} \,d^n\mathbf{q}\,{{|\mathbf{p}-\mathbf{q}|}^{ -(n+1)}}{a}^{*}(\mathbf{p})  {a}(\mathbf{q})
 ,
\end{align*}where the Fourier transformation performed in the last step can be found in \cite[Chapter III, Section 2.6]{GS1}. As for the spatial part,  essential self-adjointness of $X_{0}$ on $D(\mathbf{Q})_{\otimes}^{+}$  follows from \cite[Theorem VIII.33]{RS1}.
\end{proof}

\subsection{Translations of the Conjugate Operators}
Before deforming the conjugate operator,  we  calculate the unitary transformation of these objects  under
 translations. This is done on an operator-valued distributional level  by considering the
unitary transformation of the translations on  particle creation and annihilation operators. 
Let us first define the unitary operator of translations as follows,
\begin{equation}\label{bp}
 U(b):=e^{ib_{\mu}P^{\mu}}.
\end{equation}
The unitary operator $ U(b)$ with $b\in\mathbb{R}^{d}$ transforms the particle creation
and annihilation operators $a,a^*$ in the following way, \cite{Sch}
\begin{equation}\label{caut}
 U(b)\, {a}(\mathbf{p}) U(b)^{-1}=
e^{-ib_{\mu}p^{\mu}}{a}(\mathbf{p}), \qquad  U(b) \,{a}^{*}(\mathbf{p})
U(b)^{-1}=
e^{ i b_{\mu}p^{\mu}}a^{*} (\mathbf{p}).
\end{equation}
The transformation property (\ref{caut})  is used in the next lemma to calculate the unitary
transformations of the coordinate operator. \newline$\,$
\begin{lemma}\label{utco}
Let $V_j$ be the velocity operator  given as 
\begin{equation}\label{vel}
V_j= \int d^n\mathbf{p}\,\frac{p_{j}}{\omega_{\textbf{p}}}\,{a}^*(\mathbf{p})
{a}(\mathbf{p}).
\end{equation}
 Then, under the adjoint action of the unitary transformation $U(b):=e^{ib_{\mu}P^{\mu}}$, with
$b\in\mathbb{R}^{ {d}}$,   the  coordinate operators transform in the following way

 \begin{align}\label{x0p}\nonumber
\alpha_{ b }\left(X_{0}\right)&=  c_n
\iint
d^n \mathbf{p} \,d^n\mathbf{q}\, {e^{ib_{\mu}(p-q )^{\mu}}}{\omega_{\mathbf{p}-\mathbf{q}}^{ - d}}{a}^{*}(\mathbf{p})  {a}(\mathbf{q})
\\&=: X_0- b_\mu \widetilde{V}^{\mu}+\mathcal{O}(b^2),
\end{align} 
 \begin{align}\label{xip}
\alpha_{ b }\left(X_{j}\right)&= e^{ib_{\mu}P^{\mu}}X_{j}
e^{-ib_{\mu}P^{\mu}}=X_{j}
+b_{0}V_j
 - b_{j}N.
\end{align} 
Moreover, the spatial part of  $\widetilde{V}^{\mu}$ is the Fourier transform of the velocity operator and  the second quantization of $\mathbf{Q}\cdot |\mathbf{Q}|^{-1}$, i.e.  
\begin{equation*}
 \widetilde{\mathbf{V}} =-d\Gamma(\mathbf{Q}\cdot |\mathbf{Q}|^{-1}).
\end{equation*}

\end{lemma}
\begin{proof}
By using the transformation property (see Equation (\ref{caut})) of the creation and
annihilation operator under the unitary
operator $U(b)$, we obtain for the zero component  
\begin{align*}  
\alpha_{ b }\left(X_{0}\right)&=  c_n
\iint
d^n \mathbf{p} \,d^n\mathbf{q}\,{\omega_{\mathbf{p}-\mathbf{q}}^{ - d}}e^{ib_{\mu}p^{\mu}} {a}^{*}(\mathbf{p})  e^{-ib_{\mu}q^{\mu}} {a}(\mathbf{q})\\&=
 c_n
\iint
d^n \mathbf{p} \,d^n\mathbf{q}\,\biggl(1+ib_\mu(p-q)^\mu\biggr) {\omega_{\mathbf{p}-\mathbf{q}}^{ - d}}{a}^{*}(\mathbf{p}) {a}(\mathbf{q})+\mathcal{O}(b^2) \\&=
X_0+ ib_\mu\,c_n
\iint
d^n \mathbf{p} \,d^n\mathbf{q}\, (p-q)^\mu {\omega_{\mathbf{p}-\mathbf{q}}^{ - d}} {a}^{*}(\mathbf{p}) {a}(\mathbf{q})+\mathcal{O}(b^2)
\\&=:  X_0- b_\mu \widetilde{V}^{\mu}+\mathcal{O}(b^2).
\end{align*}
For the spatial component of the coordinate operators the adjoint action of the translations acts as follows,
\begin{align*}
e^{ib_{\mu}P^{\mu}}X_{j}
e^{-ib_{\mu}P^{\mu}}& =
-i \int d^n\mathbf{p}e^{ib_{\mu}p^{\mu}} {a}^*(\mathbf{p})
\partial_{j}\left(e^{-ib_{\mu}p^{\mu}} {a}(\mathbf{p})\right)
\\&= X_{j}
- b_{0}\int d^n\mathbf{p}(-\frac{p_{j}}{\omega_{\textbf{p}}}){a}^*(\mathbf{p})
{a}(\mathbf{p})
 - b_{k}\int d^n\mathbf{p}\,\eta_{j}^{\,\,k}
 {a}^*(\mathbf{p})
 {a}(\mathbf{p}) \\&=X_{j}
+b_{0}V_j
 - b_{j}N.
\end{align*} 
By using the Baker-Campbell-Hausdorff formula it can be seen that the adjoint action of $X_0$ gives in first oder the Fourier transform of the velocity operator, i.e. 
\begin{align*}  
\alpha_{ \mathbf{b} }\left(X_{0}\right)&= X_0+ib_j\left[P^{j}, X_0\right]+\mathcal{O}(b^2)\\&=
X_0+ib_j\left(\left[\Gamma(U_{\mathscr{F}}) X^{j}\Gamma(U_{\mathscr{F}}^{-1}) , X_0\right]\right)+\mathcal{O}(b^2)\\&=
X_0+ib_j\left(\Gamma(U_{\mathscr{F}}) \left[X^{j} ,\Gamma(U_{\mathscr{F}}^{-1}) X_0\Gamma(U_{\mathscr{F}})\right]\Gamma(U_{\mathscr{F}}^{-1})\right)+\mathcal{O}(b^2)\\&=
X_0+ib_j \left(\Gamma(U_{\mathscr{F}}) \left[X^{j} ,  P_0 \right]\Gamma(U_{\mathscr{F}}^{-1})\right)+\mathcal{O}(b^2)\\&= 
X_0-b_j \left(\Gamma(U_{\mathscr{F}})V^j\Gamma(U_{\mathscr{F}}^{-1})\right)+\mathcal{O}(b^2)\\&= 
X_0-b_j\widetilde{V}^{j}+\mathcal{O}(b^2),
\end{align*}
where in the last lines we used the unitary equivalence and the velocity operator given by the Heisenberg equation $iV^{j}=[X^j,P_0]$. Moreover,   the velocity operator in Equation (\ref{vel})  is the second quantization of $\mathbf{P}\cdot |\mathbf{P}|^{-1}$.  Hence, its Fourier transform is given as,
\begin{align*}  
\Gamma(U_{\mathscr{F}})d\Gamma(\mathbf{P} \cdot |\mathbf{P}|^{-1})\Gamma(U_{\mathscr{F}}^{-1})&= d\Gamma(U_{\mathscr{F}}\left(\mathbf{P}\cdot  |\mathbf{P}|^{-1}\right)U_{\mathscr{F}}^{-1})\\&=  -
d\Gamma(U^{-1}_{\mathscr{F}}\left(\mathbf{P}\cdot  |\mathbf{P}|^{-1}\right)U_{\mathscr{F}})\\&=-
d\Gamma(\mathbf{Q} \cdot |\mathbf{Q}|^{-1}),
\end{align*}
where in the last lines we used the properties of the Fourier transform, the properties of the second quantization and the defining Equation (\ref{uexp2}). Of course, the equivalence can as well be directly shown  on the level of the  creation and annihilation operators, as was done for $X_{\mu}$.
\end{proof}

\subsection{QFT-Moyal-Weyl from deformation}Since we work with unbounded operators,  we are  obliged to prove that the deformed product of
the conjugate operators is well-defined. In particular, the task is to establish smoothness of  the oscillatory  integral  
in a suitable locally convex topology and to demonstrate that all derivatives are polynomially bounded (see  \cite[Theorem 1]{AA} for a similar proof). 
After this is established,  we build on   results of \cite{LW} where smoothness and polynomial boundedness are the requirements needed in order to prove the well-definedness  of the warped convolutions integral.  \\\\ Now for the spatial part a simpler route will be chosen. It is easier to calculate the deformation of the spatial conjugate operator and take the commutator than to calculate the deformed product. The  two results are equivalent  since the outcome of the deformed product commutes with  the momentum operator.  This is a consequence of the fact that the product of the deformed operators   is equal to the deformation of the deformed product of the operators (see Lemma \ref{dpl}). Therefore, we consider for the spatial part   the oscillatory  integral   as follows
\begin{align*}
\langle \Psi , \mathbf{X}_{\theta}\Phi \rangle&=
(2\pi)^{-d}
\lim_{\epsilon\rightarrow 0}
  \iint  \, d x \,  d y\, e^{-ixy}  \, \chi(\epsilon x,\epsilon
y) {\langle \Psi, 
U(y)\alpha_{\theta x}(\mathbf{X})\Phi \rangle} \\&=
(2\pi)^{-d}
\lim_{\epsilon\rightarrow 0}
  \iint  \, d x \,  d y\, e^{-ixy}  \, \chi_{\theta}( \epsilon x,\epsilon
  y) {\langle \Psi, 
U(\theta y)\alpha_{  x}(\mathbf{X})\Phi \rangle} 
   ,
\end{align*}
 $\forall \Psi, \Phi \in  D(\mathbf{Q})_{\otimes}^{+}$ and  $ \chi_{\theta}( \epsilon x,\epsilon
  y):=  \chi (\epsilon\theta^{-1} x,\epsilon
\theta y)$ .
  
\begin{lemma}\label{lfpx}The deformation of the spatial conjugate operators, i.e. $\mathbf{X}_{\theta}$,  is given as a well-defined
oscillatory integral on the dense domain  $D(\mathbf{Q})_{\otimes}^{+}$.
\end{lemma}

\begin{proof}  
Following the arguments discussed above we  demonstrate the smoothness of the scalar product ${\langle \Psi, 
U(\theta y)\alpha_{  x}(\mathbf{X})\Phi \rangle}$  in $x$ and $y$ and the polynomially boundedness of all derivatives.  Hence, we start with the following expression
 \begin{align*}
 |\langle
\Psi, \partial_y^{r}\partial_x^{s}
U(\theta y)\alpha_{  x}(\mathbf{X})\Phi\rangle|=
 |\langle
\Psi, (\theta P)^r
U(\theta y)\partial_x^{s}\alpha_{  x}(\mathbf{X})\Phi\rangle|. \end{align*}
where, $r,s \in\mathbb{N}^n_0$ are multi-indices. After applying the derivative w.r.t. the variable $x$ we obtain finite linear combinations of the following terms
 \begin{align*}
 |\langle
\Psi, (\theta P)^r
U(\theta y)P^{s_1}  \alpha_{  x}(\mathbf{X})\underbrace{P^{s_2}\Phi}_{=:\Phi_s}\rangle|&
=
 |\langle  U(-\theta y)(\theta P)^rP^{s_1} 
\Psi, 
\alpha_{  x}(\mathbf{X})\Phi_s\rangle| \\&
\leq   \underbrace{\|   (\theta P)^r  P^{s_1} 
\Psi\|}_{=:c_{1,\Psi}} \| \alpha_{  x}(\mathbf{X})\Phi_s\|
 \\&=  c_{1,\Psi  }\|  \left( \mathbf{X}
+x_{0}\mathbf{V}
 - \mathbf{x}N\right)\Phi_s\| \\&\leq c_{1,\Psi }\left( \underbrace{\| \mathbf{X}\Phi_s\| }_{=:c_2}
+|x_{0}|\underbrace{\| \mathbf{V}\Phi_s\| }_{=:c_3} 
 +\| \mathbf{x}\|\underbrace{\| N\Phi_s\| }_{=:c_4} \right)\\&\leq c_{ \Psi }\left(1+|x_{0}|+\|  \mathbf{x}\|\right)
, \end{align*}
where in the last lines ${s_1} ,{s_2} \in\mathbb{N}^n_0$ are multi-indices and we used Cauchy-Schwarz.  The constant $c_{1,\Psi}$ is finite
due to the properties of the vectors $\Psi$ and $\Phi$. Moreover due to the fact that  $\Phi_s \in D(\mathbf{Q})_{\otimes}^{+}$  the constants $c_2, c_3, c_4$ are finite and a 
 constant $c_{ \Psi }$ exists such that the last inequality holds.     Now 
since  $\Phi,\,\Psi$  are  arbitrary elements
of a dense set of vectors we conclude that the deformation of the unbounded spatial conjugate operator is well-defined on $D(\mathbf{Q})_{\otimes}^{+}$.  
\end{proof}$\,$\\
For the deformed product of  the spatial and temporal conjugate operator we   consider the oscillatory  integral   as follows
\begin{align*}
\langle \Psi , (X_0\times_{\theta}\mathbf{X}) \Phi \rangle&=
(2\pi)^{-d}
\lim_{\epsilon\rightarrow 0}
  \iint  \, d x \,  d y\, e^{-ixy}  \, \chi(\epsilon x,\epsilon
y) {\langle \Psi,\alpha_{\theta x}(X_0)\alpha_{y}(\mathbf{X}) \Phi \rangle}
   ,
\end{align*}
 $\forall \Psi, \Phi \in  D(\mathbf{Q})_{\otimes}^{+}$.

 \begin{lemma}\label{lfpx1}The deformed product of  the temporal and spatial  conjugate operator, i.e. $X_0\times_{\theta}\mathbf{X}$,  is given as a well-defined
oscillatory integral on the dense domain  $D(\mathbf{Q})_{\otimes}^{+}$.
\end{lemma}
 \begin{proof}
For this proof we demonstrate that the expression in the oscillatory integral, after getting rid of the boundary terms,  is   polynomially bounded. 
Hence, we start by considering the following expression

\begin{align*} (2\pi)^{-d}
\lim_{\epsilon\rightarrow 0}
  \iint  \, d x \,  d y\, e^{-ixy}  \, \chi(\epsilon x,\epsilon
y) 
\langle \Psi  ,\alpha_{\theta x}(X_0) \left(\mathbf{X}
+y_{0}\mathbf{V}
 -\mathbf{y}N\right)  { \Phi } \rangle     ,
\end{align*}
where in the last line we used the adjoint action of the translations on $\mathbf{X}$ (see Equation (\ref{xip})).
Since the adjoint action of the momentum operator on $X_0$   is from a calculational point of view hard to handle to all orders in $(\theta x)$, we mention that  we only need to look at the expansion in first order. This is due to the fact that all other orders vanish since   the terms in $y$ are only of first order. For a proof of this statement review \cite[Proof of Lemma 5.3]{MUc}. However, this fact can   be easily seen by rewriting the terms in $(\theta x)$ as partial derivatives of the exponential term. After partial integration, terms of unequal order in $x$ and  $y$ vanish since derivatives of  $\chi(\epsilon x,\epsilon
y)$ contain powers of $\epsilon$ which vanish in the limit $\epsilon\rightarrow 0$.
Now by using adjoint action of $X_0$ (see Equation (\ref{x0p})) we obtain   for the deformed product the following,
\begin{align*}&(2\pi)^{-d}
\lim_{\epsilon\rightarrow 0}
  \iint  \, d x \,  d y\, e^{-ixy}  \, \chi(\epsilon x,\epsilon
y)\langle  \Psi , (X_0- (\theta x)_\mu \widetilde{V}^{\mu}) \left(\mathbf{X}
+y_{0}\mathbf{V}
 -\mathbf{y}N\right)\Phi \rangle \\&=(2\pi)^{-d}
\lim_{\epsilon\rightarrow 0}
  \iint  \, d x \,  d y\, e^{-ixy}  \, \chi(\epsilon x,\epsilon
y)\langle  \Psi ,  \biggl(
X_0 \mathbf{X}-y_0(\theta x)_\mu \widetilde{V}^{\mu} \mathbf{V}  
+ \mathbf{y}(\theta x)_\mu \widetilde{V}^{\mu}  N
 \biggr)\Phi \rangle\\&= (2\pi)^{-d}
\lim_{\epsilon\rightarrow 0}
  \iint  \, d x \,  d y\, e^{-ixy}  \, \chi(\epsilon x,\epsilon
y)\underbrace{\langle  \Psi ,  \biggl(
X_0 \mathbf{X}+ i \theta  _{\mu 0} \widetilde{V}^{\mu}  \mathbf{V}   
-i \theta_{\mu j}   \widetilde{V}^{\mu} N
 \biggr)\Phi \rangle}_{\leq \|\Psi\|\left(\|X_0 \mathbf{X}\Phi\|+\| \theta  _{\mu 0} \widetilde{V}^{\mu}  \mathbf{V}   \Phi\|+\| \theta_{\mu j}   \widetilde{V}^{\mu} N\Phi\|\right)=: C_{\Psi,\Phi} }  \\&\leq C_{\Psi,\Phi}(2\pi)^{-d} 
\lim_{\varepsilon_1\rightarrow 0}  \left(
\int dx \lim_{\varepsilon_2\rightarrow 0} 
\left(\int dy  e^{-ixy}
\chi_2(\varepsilon_2 y)\right)\,\chi_1(\varepsilon_1  x)\,
 \right)
\\
&=   C_{\Psi,\Phi}    (2\pi)^{-d/2}
\lim_{\varepsilon_1\rightarrow 0}  \left(
\int dx \,
\delta( {x} )\,\chi_1(\varepsilon_1  x)    \right) 
 =  C_{\Psi,\Phi}     .
\end{align*}
Here we used the fact that the
oscillatory integral does not depend on the cut-off function  chosen. As in
\cite{RI}, we chose $\chi (\varepsilon k,\varepsilon y)= \chi_2(\varepsilon_2 k
)\chi_1(\varepsilon_1 y)$ 
with $\chi  \in\mathscr{S}(\mathbb{R}^{d}\times\mathbb{R}^{d})$ and $\chi_{l}(0 )=1$, $l=1$, $2$,
and obtained the delta
distribution $\delta(y)$ in the limit $\varepsilon_2 \rightarrow
0$, \cite[Section 7.8, Equation 7.8.5]{H}.   
Hence, the remaining term inside the oscillatory integral is bounded by a constant and the integral is well-defined. 
  Now the scalar product under the integral is  estimated by using Cauchy-Schwarz. The finiteness of $C_{\Psi,\Phi} $ follows from the property of the vectors 
$ \Psi$ and $\Phi$. Equivalent considerations can be done for the deformed product $\mathbf{X}\times_{\theta}X_0$ and hence the deformed commutator of the conjugate operators is given as a well-defined oscillatory integral. 
 \end{proof}
\begin{remark}
In addition to the polynomial boundedness, smoothness w.r.t. the variables $x$ and $y$ is a necessary requirement. However,  derivatives   produce   finite linear combinations of the following terms  $
\langle (\theta P)^{r_1}\Psi  ,\alpha_{\theta x}(X_0)(\theta P)^{r_2} P^{s_1}\left(\mathbf{X}
+y_{0}\mathbf{V}
 -\mathbf{y}N\right)  {P^{s_1} \Phi } \rangle$. Hence, we follow the same arguments as before and end up with a similar expression with the exception of additional polynomials of the momentum operator. Now by estimating the term using Cauchy-Schwarz and by using the fact that the vectors $ \Psi, \Phi \in  D(\mathbf{Q})_{\otimes}^{+}$ finiteness of the resulting constant is guaranteed.
\end{remark}
$\,$\\
 Next we turn to the actual result of the  deformed product of the  coordinate operators.
\newline
\begin{lemma}\label{qftdefx}
The deformed  coordinate operator   $X^{j}_{\theta}$, obtained by warped convolutions,  represented on the dense domain $D(\mathbf{Q})_{\otimes}^{+}$ is given by
\begin{equation*} 
X^{j}_{\theta}=X^{j} +\left(\theta P \right)^{0}V^j - \left(\theta P \right)^{j}N.
\end{equation*}
\end{lemma}
\begin{proof}
We use the defining equation (\ref{WC}) and Lemma \ref{utco}.   The  calculation can be performed  for the spatial coordinate operator $X_{j}$ on
$\Psi \in
D(\mathbf{Q})_{\otimes}^{+}$ as follows,
\begin{align*}
  X^{j}_{\theta}\Psi &= \int\alpha_{\theta k}( X^{j}) dE(k)\Psi  
 \\&= \int \left(X^{j}
+(\theta k)_{0}V^j
 -(\theta k)^{j}N
\right)dE(k)\Psi  
 \\&=  \left(X^{j} +\left(\theta P \right)^{0}V^j - \left(\theta P \right)^{j} N\right)
\Psi.
\end{align*}
\end{proof} 
In the next step we give the commutator of the deformed coordinate operator. In order to make  relativistic corrections more apparent   we do not set the speed of light $c$ equal to one.  \newline
\begin{theorem}\label{qmwxc} \textbf{QFT-Moyal-Weyl.} The deformed commutator of the coordinate operators represented on the dense domain $D(\mathbf{Q})_{\otimes}^{+}$ is given by
\begin{align*} 
[X_\mu\stackrel{\times_{\theta}}{,} X_\nu] &=   i\widehat{\theta}_{\mu\nu},
\end{align*}
where the operator-valued matrix  $ i\widehat{\theta}$ is in the $0j$-component given as 
\begin{align*} 
  i\widehat{\theta}_{0j}=  i    \theta_{0k}  \, \widetilde{V} ^{k}_{\,\,j}
-2i\theta_{jk}  \,
d\Gamma( {Q}^k \cdot |\mathbf{Q}|^{-1})N ,
\end{align*}
and in the spatial section, i.e. in the $ij$-components given as 
\begin{align*} 
  i\widehat{\theta}_{ij}=  -2 i\left(
\theta_{0i} V_j /c-\theta_{0j} V_i/c \right)N
-2i\theta_{ij}N^2.
\end{align*}
The operator $\widetilde{V}^{k}_{\,\,j}$ is the observable given by the commutator of the velocity with the Fourier transformed velocity, i.e. $ \widetilde{V}^{k}_{\,\,\,j} =[ \widetilde{V}^{k},V_j]$.
\end{theorem}
\begin{proof}
To calculate the commutator of the deformed product we first calculate the deformed product of $X_0$ and $X_j$ on  $\Psi  \in D(\mathbf{Q})_{\otimes}^{+}$,
\begin{align*} 
(X_0\times_{\theta}X_j)\Psi&=(2\pi)^{- d} 
\iint dx\,dy \, e^{-ixy} \alpha_{\theta x}(X_0)\alpha_{y}(X_j)\Psi\\&=
(2\pi)^{- d} 
\iint dx\,dy \, e^{-ixy} \alpha_{\theta x}(X_0)\left(X_{j}
+y_{0}V_j
 - y_{j}N\right)\Psi \\
 &=
(2\pi)^{- d} 
\iint dx\,dy \, e^{-ixy}  (X_0- (\theta x)_\mu \widetilde{V}^{\mu} )\left(X_{j}
+y_{0}V_j
 - y_{j}N\right)\Psi \\&=(2\pi)^{- d} 
\iint dx\,dy \, e^{-ixy}  \left(X_0X_{j}- (\theta x)_\mu y_{0} \widetilde{V} ^{\mu}V_j
+(\theta x)_\mu  y_{j}\widetilde{V} ^{\mu}N\right)\Psi  \\&=
X_0X_j\Psi + (2\pi)^{- d} 
\iint dx\,dy \, e^{-ixy}  \left( i\frac{\partial}{\partial x^0}(\theta x)_\mu   \widetilde{V} ^{\mu}V_j
-i\frac{\partial}{\partial x^j}(\theta x)_\mu \widetilde{V} ^{\mu} N\right)\Psi 
 \\&=X_0X_j\Psi +i (2\pi)^{- d} 
\iint dx\,dy \, e^{-ixy}  \left( \theta_{\mu0}  \, \widetilde{V} ^{\mu}V_j
- \theta_{\mu j}\, \widetilde{V} ^{\mu} N\right)\Psi  \\&=\left(X_0X_j  +i    \theta_{\mu 0}  \, \widetilde{V} ^{\mu}V_j
-i \theta_{\mu j}\, \widetilde{V} ^{\mu} N\right)\Psi ,
\end{align*}
where in the last lines we used the fact that the only nonzero terms in the expansion are  the terms of
equal order in $(\theta x)$ and $y$. Moreover, we rewrote $y_{\mu}$ as a derivative on the exponential and performed a partial integration. 
To calculate the commutator we additionally need the deformed product in exchanged order, i.e. 
\begin{align*} i&(2\pi)^{- d} 
\iint dx\,dy \, e^{-ixy} \alpha_{\theta x}(X_j)\alpha_{y}(X_0)\Psi\\&=
(2\pi)^{- d} 
\iint dx\,dy \, e^{-ixy}
\left(X_{j}
+(\theta x)_{0}V_j
 - (\theta x)_{j}N\right)\left(X_0-y_\mu \widetilde{V} ^{\mu}\right)\Psi \\&=(2\pi)^{- d} 
\iint dx\,dy \, e^{-ixy}  \left(X_jX_{0}- (\theta x)_0 y_{\mu} V_j \widetilde{V} ^{\mu}
+(\theta x)_j  y_{\mu}N \widetilde{V} ^{\mu}\right)\Psi \\&=X_jX_{0} \Psi +(2\pi)^{- d} 
\iint dx\,dy \, e^{-ixy}  \left( i  \frac{\partial}{\partial x^\mu}(\theta x)_0  V_j \widetilde{V} ^{\mu}
-i \frac{\partial}{\partial x^\mu}(\theta x)_j   N \widetilde{V} ^{\mu}\right)\Psi \\&=X_jX_{0}\Psi  +i(2\pi)^{- d} 
\iint dx\,dy \, e^{-ixy}  \left( \theta_{0 \mu}  V_j \widetilde{V} ^{\mu}
-\theta_{j\mu}  N \widetilde{V} ^{\mu}\right)\Psi \\&=\left( X_jX_{0}   +i
  \theta_{0 \mu}  V_j \widetilde{V} ^{\mu}
-i\theta_{j\mu}  N \widetilde{V} ^{\mu}\right)\Psi .
\end{align*}
By using the last two calculations we obtain the  commutator of the deformed product which is given as,
\begin{align*} 
[X_0\stackrel{\times_{\theta}}{,} X_j] \Psi&= 
\left([X_0,X_j]  +i    \theta_{\mu 0}  \, [\widetilde{V} ^{\mu},V_j]
-i\theta_{\mu j}\, \widetilde{V} ^{\mu} N
+i\theta_{j\mu}  N \widetilde{V} ^{\mu}
\right)\Psi\\&= 
\left([X_0,X_j]  +i    \theta_{\mu 0}  \, [\widetilde{V} ^{\mu},V_j]
+i\theta_{0 j}\, [N,\widetilde{V} ^{0} ] 
+i\theta_{jk}  \{N, \widetilde{V} ^{k}\}
\right)\Psi,
\end{align*}
where in the last lines we used the  skew-symmetry of $\theta$ w.r.t. the Minkowski metric, i.e.
$\theta_{0j}=\theta_{j0}$, $\theta_{kj}=-\theta_{jk}$ and $\{. \,,\, . \}$ denotes the anticommutator.
Now the first commutator vanishes since,
\begin{align*} 
[X_0,X_j]=\Gamma(U^{-1}_{\mathscr{F}})\underbrace{[P_0,P_j]}_{=0}\Gamma(U_{\mathscr{F}})=0.
\end{align*}
Also note that $N$ commutes with all particle number conserving operators and since $\widetilde{V}$ is of the form $a^*a$ it commutes with $N$. Nevertheless, this statement can also be easily proven by using the canonical commutation relations (\ref{pccr}) and the form of the particle number operator  given in Equation (\ref{pcaopm}). Hence the only term that remains is the second one, and we calculate it in the following,
\begin{align*} [\widetilde{V} ^{\mu},V_j]&
=-i\iiint 
d^n \mathbf{p} \,d^n\mathbf{q}\,d^n\mathbf{k}\left(\frac{ (p-q)^\mu}{{\omega_{\mathbf{p}-\mathbf{q}}^{   d}}} \frac{k_{j}}{\omega_{\textbf{k}}}\right)\underbrace{ [{a}^{*}(\mathbf{p})   {a}(\mathbf{q}),   {a}^*(\mathbf{k})
{a}(\mathbf{k})]}_{-
\delta^n( \mathbf{p}-\mathbf{k})
{a}^*(\mathbf{k})
{a}(\mathbf{q})+\delta^n( \mathbf{q}-\mathbf{k})
{a}^{*}(\mathbf{p}){a}(\mathbf{k})
}\\&=-i
\iint 
d^n \mathbf{p} \,d^n\mathbf{q}
\left(\frac{ (p-q)^\mu}{{\omega_{\mathbf{p}-\mathbf{q}}^{   d}} }\left( \frac{q_{j}}{\omega_{\textbf{q}}}
-\frac{p_{j}}{\omega_{\textbf{p}}}
\right)
\right)
{a}^{*}(\mathbf{p}){a}(\mathbf{q})=:\widetilde{V} ^{\mu}_{\,\,j}.
\end{align*}
Hence by summing up all the commutators we obtain  
\begin{align*} 
[X_0\stackrel{\times_{\theta}}{,} X_j] \Psi&= \biggl(
  i    \theta_{\mu 0}  \, \widetilde{V} ^{\mu}_{\,\,j}
+2i\theta_{jk}   \widetilde{V} ^{k} N \biggr)\Psi\\&= \biggl(
   i    \theta_{0k}  \, \widetilde{V} ^{k}_{\,\,j}
+2i\theta_{jk}   \widetilde{V} ^{k} N \biggr)\Psi.
 \end{align*}
Next we turn our attention to the commutator of the deformed product of the spatial coordinate operators.  
 We use the algebra of the massless coordinate operators, the algebra of the momentum operators
and the fact that the particle number operator $N$ commutes with $X^{\mu}$ and $P^{\nu}$ to
calculate the commutator  of the deformed spatial operator. 
\begin{align*}
 [X_{i}^{\theta},X_{j}^{\theta}]\Psi&=   [ X_{i}, \left(\theta P \right)_{0}V_j]\Psi-[ X_{i}, \left(\theta P \right)_{j}N]\Psi-
i\leftrightarrow j
\\&=   \left(\theta P \right)_{0}[ X_{i}, V_j]\Psi+ [ X_{i}, \left(\theta P \right)_{0}]V_j\Psi
-[ X_{i}, \left(\theta P \right)_{j}]N\Psi-
i\leftrightarrow j
\\&  = -2
i\left(
\theta_{0i} V_j -\theta_{0j} V_i \right)N\Psi
-2i\theta_{ij}N^2\Psi
\end{align*}
In the last lines we used  the skew-symmetry of $\theta$ w.r.t. the
Minkowski metric. Note that the commutator of the deformed product of the spatial part is equal to the commutator of the deformed operator 
$\mathbf{X}_{\theta}$  since, 
 \begin{align*}
\left( [X^{j}_{\theta},X^{k}_{\theta}]\right)_{-\theta }\Psi=   [X^j\stackrel{\times_{\theta}}{,} X^k] \Psi,
\end{align*}
and the result of  $[X^{j}_{\theta},X^{k}_{\theta}]$ commutes with $\mathbf{P}$. To see this we only have to calculate the commutator of the momentum operator with $\mathbf{V}$ since $N$ commutes with $\mathbf{P}$.
 \begin{align*}
  [V^{j} ,P^{\mu} ]& = 
\iint d^n\mathbf{p}d^n\mathbf{q}\, \frac{p_{j}}{\omega_{\textbf{p}}}\,q^{\mu}\underbrace{ [{a}^{*}(\mathbf{p})   {a}(\mathbf{p}),   {a}^*(\mathbf{q})
{a}(\mathbf{q})]}_{-
\delta^n( \mathbf{p}-\mathbf{q})
{a}^*(\mathbf{q})
{a}(\mathbf{p})+\delta^n( \mathbf{p}-\mathbf{q})
{a}^{*}(\mathbf{p}){a}(\mathbf{q})
}\\&= 
\int d^n\mathbf{p}\, \left(-\frac{p_{j}}{\omega_{\textbf{p}}}\,p^{\mu}+\frac{p_{j}}{\omega_{\textbf{p}}}\,p^{\mu}\right){a}^{*}(\mathbf{p})   {a}(\mathbf{p})=0.
\end{align*}
In the last lines we used the canonical commutation relations (\ref{pccr})  and therefore we obtain,
 \begin{align*}
\left( [X^{j}_{\theta},X^{k}_{\theta}]\right)_{-\theta }\Psi=  [X^{j}_{\theta},X^{k}_{\theta}]\Psi=[X^j\stackrel{\times_{\theta}}{,} X^k]\Psi. 
\end{align*}
\end{proof}
It turns out that the relativistic correction terms appearing, for the spatial part, are of the order $v/c$ and hence not equal  to the standard  Moyal-Weyl. Even more striking is the commutator relation between the zero component and the spatial component. It depends in both terms on the Fourier transformed velocity and even on $\theta_{ik}$. This  suggests that setting $\theta_{0k}$ equal to zero, as  often done in order to simplify calculations, is no longer an option. Of course, it is  important to point out that in the one-particle and  nonrelativistic limit, i.e. $v/c<<1$, we obtain the standard Moyal-Weyl spacetime for the spatial section. 
 
\section{Conclusion and Outlook}
 
We obtained the QFT-Moyal-Weyl by taking the  route layed out in a QM-context. In QM the deformation of the coordinate operator gives us, after the identification of $\theta$ with the inverse of a magnetic field, the so called guiding center
coordinates. These coordinates  satisfy the commutator relations of the Moyal-Weyl Plane. In the Landau problem this plane is not merely
an abstract construct but has the physical meaning, that the space coordinates can not be
measured simultaneously.\\\\ Now at least the spatial part of the  QFT-Moyal-Weyl plane can be viewed as a  first order relativistic correction and second quantization of the Landau plane, after making the appropriate identifications, i.e. $\theta_{\mu \nu}=F^{-1}_{\mu \nu}$ for a constant electromagnetic field strength tensor $F$. This fact is first of all supported by working in  a relativistic second-quantized context. Secondly, in the one particle non-relativistic limit we obtain the well known Landau plane (see \cite{ Muc1}). Hence, the QFT - noncommutative plane seems to be an intermediate stepping stone from the non-relativistic one-particle Landau quantization to the second quantized relativistic one. However, the case of a fully relativistic Landau effect and an emerging quantum spacetime thereof should be studied in more detail. In this context deformation theory seems to be a suitable approach.
\\\\
The commutator relations of the temporal and spatial part shed a new light on the Moyal-Weyl plane. First of all, the $\theta_{0j}$ term comes with an observable which is the commutator of the velocity with the Fourier-transformed velocity. This observable is not equal to the Kronecker-delta, as one can easily see from the form of the expression. In a sense the observable resembles the canonical commutation relations due to the definitions of the velocities in the different spaces. Nevertheless, we are not aware of any appearance of such a term in  literature. Moreover, we have a second term in the commutator relations of $\hat{\theta}_{0j}$. This second term resembles the second quantization of a Lie-algebraic structure and is closely related to noncommutative spacetimes dubbed $\kappa$-Minkowski, \cite{AJ}. This can be seen clearer if we set $\theta_{0j}=0$ and apply the commutator relations on square-integrable   functions defined  on a sphere. Hence, turning the norm of $\mathbf{Q}$ into a constant  identified with the radius of the sphere.
\\\\
The second-quantization also plays a role in the appearance of the particle number operator. It is not clear if this is an integral part of the QFT-Moyal Weyl or merely an effect of the deformation. The outcome is connected to the strategy we have taken. If we would have first deformed our conjugate operators and then performed a second quantization we would not have any particle number appearing terms, except for the $N^2$-term which would be linear in $N$.  However, one could also define the conjugate operators with an $N^{-1}$ as was done in \cite{BDFP} for the mean event coordinate operator. The change will only be visible at the level of the commutator relations since $N$ commutes with all particle number preserving operators. The constant part of the spatial  QFT-Moyal Weyl would have no particle number appearing and hence not differ from the conventional Moyal-Weyl. Moreover, the velocity dependent term would  have the inverse of the particle number operator. Since all these possibilities have their theoretical  justification only an experiment or arguments from renormalization theory can settle the uniqueness of the QFT-Moyal Weyl.
\\\\
Our intuition in using the massless scalar field is guided by simplicity. Nevertheless, by using the massive free scalar field the commutation relations do not change their form. The only obvious exception are the velocity and the Fourier transformed velocity, which would be interchanged with their massive counterparts. However, in the massive case one cannot draw analogies to the $\kappa$-Minkowski spacetime.
\\\\
Since, the   QFT-Moyal-Weyl spacetime has some new features as the incorporation of the velocity and the particle number it would be interesting to study it  in the context of scattering.

 \bibliographystyle{alpha}
\bibliography{allliterature}

 \end{document}